\documentclass[preprint]{elsarticle}

\usepackage{array,xspace,multirow,hhline,tikz,colortbl,tabularx,booktabs,fixltx2e,amsmath,amssymb,amsfonts,amsthm}
\usepackage{algorithm}
\usepackage{algorithmic}
\usepackage{verbatim,ifthen}
\usepackage{pifont}
\usepackage{ifthen}
\usepackage{calrsfs,mathrsfs}
\usepackage{bbding,pifont}
\usepackage{pgflibraryshapes}
\usepackage{paralist}
\usepackage{balance}

\usepackage{mathtools}
\usepackage{amssymb,amsfonts,amsthm,xspace}
\usepackage{booktabs}
\usepackage{paralist}
\usepackage{tikz}
\usepackage{nicefrac}
\usepackage{nicefrac}
\usepackage{url}

	\usepackage{varioref}
	
	\usetikzlibrary{arrows}
	\usetikzlibrary{positioning}
	\usetikzlibrary{calc}
	\usetikzlibrary{shapes,decorations}

\definecolor{light-gray}{gray}{0.9}

\newtheorem{definition}{Definition}%

	\newcommand{\eg}{e.g.,\xspace}

	\newtheorem{lemma}{Lemma}%
	\newtheorem{theorem}{Theorem}%
	\newtheorem{proposition}{Proposition}%
	\newtheorem{corollary}{Corollary}%
	\newtheorem{example}{Example}

\newtheorem{remark}[definition]{Remark}

	\newcommand\eat[1]{}

	\usepackage{enumitem}
	\setenumerate[1]{label=\rm(\it{\roman{*}}\rm),ref=({\it\roman{*}}),leftmargin=*}
	\newlength{\wordlength}
	\newcommand{\wordbox}[3][c]{\settowidth{\wordlength}{#3}\makebox[\wordlength][#1]{#2}}

	\newcommand{\set}[1]{\{#1\}}
	\newcommand{\midd}{\mathbin{:}}

	\tikzset{
		inner sep=0pt, outer sep=0pt, minimum size=0pt, thick,
		level/.style={sibling distance = (\columnwidth/16)*2^(4-#1)},
		winner/.style={minimum size=1.5em, circle, draw, fill=white, font={\footnotesize}},
		leaf/.style={inner sep=.15em, font={\footnotesize}},
		ball/.style={minimum size=.4em,circle,fill=black},
		beats/.style={thick,->,>=stealth',draw}
	}

	\newcommand{\pref}{\succ \xspace}

	\newcommand{\Indiff}[1][]{
		\ifthenelse{\equal{#1}{}}{\mathrel I}{\mathop{I_{#1}}}
	}
	\newcommand{\prefset}[1][]{\ifthenelse{\equal{#1}{}}{\mathcal{R}}{\mathcal{R}_{#1}}}

\begin{document}

\title{Manipulating the Probabilistic Serial Rule}

	\author{Haris Aziz\corref{cor1}} \ead{haris.aziz@nicta.com.au}
	\author{Serge Gaspers} \ead{serge.gaspers@nicta.com.au}
	\author{Simon Mackenzie} \ead{Simon.Mackenzie@nicta.com.au}
	\author{Nicholas Mattei} \ead{Nicholas.Mattei@nicta.com.au}
		\address{NICTA and UNSW, Kensington 2033, Australia}
	\author{Nina Narodytska} \ead{ninan@cs.cmu.edu}
	\address{Carnegie Mellon University, Pittsburgh, PA 15213-3891, USA}
	\author{Toby Walsh} \ead{toby.walsh@nicta.com.au}
	\address{NICTA and UNSW, Kensington 2033, Australia}



\begin{abstract}
	The probabilistic serial (PS) rule is one of the most prominent randomized rules for the assignment problem. It is well-known for its superior fairness and welfare properties. However, PS is not immune to manipulative behaviour by the agents. 
	We initiate the study of the computational complexity of an agent manipulating the PS rule. 
	We show that computing an expected utility better response is NP-hard. 
	On the other hand, we present a polynomial-time algorithm to  compute a lexicographic  best response. 
	For the case of two agents, we show that even an expected utility best response can be computed in polynomial time. Our result for the case of two agents relies on an interesting connection with sequential allocation of discrete objects. 
\end{abstract}


\maketitle

		\section{Introduction}

		The \emph{assignment problem} is one of the most fundamental and important problems in economics and computer science~\citep[see \eg][]{ABS13a,AGMW14a,BoMo01a,Gard73b,HyZe79a}. 
		In the setting, agents express preferences over objects and, based on these preferences, the objects are allocated to the agents. The model is applicable to many resource allocation or fair division settings where the objects may be public houses, school seats, course enrollments, kidneys for transplant, car park spaces, chores, joint assets, or time slots in schedules. 
		A randomized or fractional assignment rule takes the preferences of the agents into account in order to allocate each agent a fraction of the object. If the objects are indivisible but allocated in a randomized way, the fraction can also be interpreted as the probability of receiving the object. Randomization is widespread in resource allocation since it is one of the most natural ways to ensure procedural fairness~\citep{BCKM12a}. Randomized assignments have been used to assign public land, radio spectra to broadcasting companies, and US permanent visas to applicants~\citep[Footnote~1, ][]{BCKM12a}.

Among the various randomized/fractional assignment rules, the \emph{probabilistic serial (PS)} rule is one of the most prominent rules~\citep{AzSt14a,BoHe12a,BoMo01a,BCKM12a, KaSe06a,Koji09a, Yilm10a,SaSe13b}.  PS works as follows. 
Each agent expresses a linear order over the set of houses (we use the term house throughout the paper though we stress any object could be allocated with these mechanisms). Each house
		is considered to have a divisible probability weight of one, and agents simultaneously and with the same speed eat the probability weight of their most preferred house. Once a house has been eaten by a subset of agents, these agents proceed to eat their next most preferred house that has not been completely eaten. The procedure terminates after all the houses have been eaten. The random allocation of an agent by PS is the amount of each object he has eaten. Although PS was originally defined for the setting where the number of houses is equal to the number of agents, it can be used without any modification for fewer or more houses than agents~\citep[see \eg][]{BoMo01a,Koji09a}.

		The \emph{probabilistic serial (PS)} rule fares better  than any other random assignment rule in terms of fairness and welfare~\citep{BoHe12a,BoMo01a,BCKM12a,Koji09a,Yilm10a}.
		In particular, it satisfies strong envy-freeness and efficiency with respect to both \emph{stochastic dominance (SD)} and \emph{downward lexicographic (DL)} relations~\citep{BoMo01a,ScVa12a,Koji09a}.
		  SD is one of the most fundamental relations between fractional allocations because one allocation is SD-preferred over another if for every utility function consistent with the ordinal preferences, the former yields at least as much expected utility as the latter. DL is a refinement of SD and based on lexicographic comparisons between fractional allocations.
		Generalizations of the PS rule have been recommended in many settings~\citep[see \eg][]{BCKM12a}.
		The PS rule also satisfies some desirable incentive properties. If the number of objects is at most the number of agents, then PS is weak SD-strategyproof~\citep{BoMo01a}. 
		Another well-established rule \textit{random serial dictator (RSD)} is not envy-free, not as efficient as PS~\citep{BoMo01a} and the fractional allocations under RSD  are \#P-complete to compute~\citep{ABB13b}. However, unlike RSD,  PS is not strategyproof.


In this paper, 	we examine the following natural question for the first time: \emph{what is the computational complexity of an agent computing a different preference to report so as to get a better PS outcome?} This problem of computing the optimal manipulation has already been studied in great depth for voting rules~\citep[see \eg][]{FHH10a}.
		  \citet{EkKe12a} showed that when agents are not truthful, the outcome of PS may not satisfy desirable properties related to efficiency and envy-freeness. 
		Hence, it is important to check that even if agents can in principle manipulate, how hard it is to compute a beneficial misreport of their preferences. The complexity of manipulation  of the PS rule is also related to the study of Nash dynamics and better responses. Efficient algorithms to compute best responses can be used to understand Nash dynamics under the mechanism.
		

		In order to compare random allocations, an agent needs to consider relations between them. We consider three well-known relations between random allocations~\citep[see \eg][]{BoMo01a,ScVa12a,SaSe13b,Cho12a}:
		$(i)$ \textit{expected utility (EU)}, $(ii)$ \textit{stochastic dominance (SD)},  and $(iii)$ \textit{downward lexicographic (DL)}.
		For EU, an agent seeks a different allocation that yields more expected utility. 
		For DL, an agent seeks an allocation that gives a higher probability to the most preferred alternative that has different probabilities in the two allocations.
		Throughout the paper, we assume that agents express \emph{strict} preferences, i.e., they are not indifferent between any two houses. 



		\textbf{Contributions. }We initiate the study of computing best responses for the PS mechanism --- one of the most established randomized rules for the assignment problem. 
		The study is additionally motivated by complementing experimental work where we observe that as the number of houses relative to the number of agents grows, the percentage of  manipulable profiles (for which at least one agent has incentive to manipulate) increases, maximizing at around $99\%$. 
		We present a polynomial-time algorithm to compute the DL best response for multiple agents and houses. 
		For the case of two agents, we present a polynomial-time algorithm to compute an EU best response for any utilities consistent with the ordinal preferences. The two-agent case is also of special importance since various disputes arise between two parties.
		The result for the EU best response relies on an interesting connection between the PS rule and the sequential allocation rule for indivisible objects. In a sequential allocation, a picking sequence is specified for the agents and agent get his most preferred available object when his turns comes. For general $n$, we show that computing an  EU best response is NP-hard. The result contrasts sharply with the recent result of \citet{BoLa14a} that a best response can be computed in polynomial time for sequential allocation.

		\section{Preliminaries}


		An \emph{assignment problem} $(N, H, \pref)$ consists  of a set of agents $N=\{1,\ldots, n\}$, a set of houses $H=\{h_1, \ldots, h_m\}$ and a preference profile $\pref=(\pref_1,\ldots, \pref_n)$ in which $\pref_i$ denotes a complete, transitive and strict ordering on $H$ representing the preferences of agent $i$ over the houses in  $H$.
			A \emph{fractional assignment} is an $(n\times m)$ matrix $[p(i)(h_j)]_{\substack{1\leq i\leq n, 1\leq j\leq m}}$ such that for all $i\in N$, and $h_j\in H$, $0\leq p(i)(h_j)\leq 1$;  and for all $j\in \{1,\ldots, n\}$, $\sum_{i\in N}p(i)(h_j)= 1$. 
		The value $p(i)(h_j)$ is the fraction of house $h_j$ that agent $i$ gets. Each row $p(i)=(p(i)(h_1),\ldots, p(i)(h_m))$ represents the \emph{allocation} of agent $i$.
	A fractional assignment can also be interpreted as a random assignment where $p(i)(h_j)$ is the probability of agent $i$ getting house $h_j$.

			A standard method to compare random allocations is to use the \emph{SD (stochastic dominance)} relation. 
			%
			 Given two random assignments $p$ and $q$, $p(i) \succ_i^{SD} q(i)$ i.e.,  a player $i$ \emph{SD~prefers} allocation $p(i)$ to $q(i)$ if
			$\sum_{h_j\in \set{h_k\midd h_k\pref_i h}}p(i)(h_j) \ge \sum_{h_j\in \set{h_k\midd h_k\pref_i h}}q(i)(h_j)$  for all  $h\in H$ and 
			$\sum_{h_j\in \set{h_k\midd h_k\pref_i h}}p(i)(h_j) > \sum_{h_j\in \set{h_k\midd h_k\pref_i h}}q(i)(h_j) \text{ for some } h\in H.$

			Given two random assignments $p$ and $q$, $p(i) \pref_i^{DL} q(i)$ i.e.,  a player $i$ \emph{DL~(downward lexicographic) prefers}  allocation $p(i)$ to $q(i)$ if $p(i)\neq q(i)$ and for the most preferred house $h$ such that $p(i)(h)\neq q(i)(h)$, we have that $p(i)(h)>q(i)(h)$.

				When agents are considered to have cardinal utilities for the objects, we denote by $u_i(h)$ the utility that agent $i$ gets from house $h$. We will assume that the total utility of an agent equals the sum of the utilities that he gets from each of the houses. Given two random assignments $p$ and $q$, $p(i) \pref_i^{EU} q(i)$ i.e.,  a player $i$ \emph{EU (expected utility)~prefers} allocation $p(i)$ to $q(i)$ if
			$\sum_{h\in H}u_i(h) \cdot p(i)(h)> \sum_{h\in H}u_i(h) \cdot q(i)(h).$

		Since for all $i\in N$, agent $i$ compares assignment $p$ with assignment $q$ only with respect to his allocations $p(i)$ and $q(i)$, we will sometimes abuse the notation and use $p\pref_i^{SD} q$ for $p(i)\pref_i^{SD} q(i)$.
		A \emph{random assignment rule} takes as input an assignment problem $(N,H,\pref)$ and returns a random assignment which specifies what fraction or probability of each house is allocated to each agent.


		%

		\section{The Probabilistic Serial Rule and its Manipulation}
The \emph{Probabilistic Serial (PS) rule} is a random assignment algorithm in which we consider each house as infinitely divisible~\citep{BoMo01a,Koji09a}.
		At each point in time, each agent is eating (consuming the probability mass of) his most preferred house that has not been completely eaten and each agent eats at the same unit speed. Hence all the houses are eaten at time $m/n$ and each agent receives a total of $m/n$ units of houses.
		The probability of house $h_j$ being allocated to $i$ is the fraction of house $h_j$ that $i$ has eaten.
 The following example adapted from \citep[Section 7, ][]{BoMo01a} shows how PS works.

		\begin{example}[PS rule]\label{example:PS}
			Consider an assignment problem with the following preference profile.
		\begin{align*}
			\succ_1:\quad& h_1,h_2,h_3 & \succ_2:\quad& h_2,h_1,h_3&	\succ_3:\quad& h_2,h_3,h_1
			\end{align*}
			Agents $2$ and $3$ start eating $h_2$ simultaneously whereas agent $1$ eats $h_1$. When $2$ and $3$ finish $h_2$, agent $1$ has only eaten half of $h_1$.  The timing of the eating can be seen below.
		\begin{center}
		             \begin{tikzpicture}[scale=0.2]
		                 \centering
		                 \draw[-] (0,0) -- (0,6);
		                 \draw[-] (0,0) -- (20,0);

		                 \draw[-] (20,6) -- (20,0);

		\draw[-] (0,2) -- (20,2);
		\draw[-] (0,4) -- (20,4);
		\draw[-] (20,0) -- (20,6);

		\draw[-] (10,0) -- (10,6);

		\draw[-] (0,6) -- (20,6);

		\draw[-] (15,0) -- (15,6);

		                                        \draw (0,-.8) node(c){\small $0$};
		                             \draw (20/2,-1.2) node(c){\small $\frac{1}{2}$};

		 \draw (20/2,-2.5) node(c){\small Time};

		                             \draw (20,-1) node(c){\small$1$};

		\draw (15,-1.2) node(c){\small$\frac{3}{4}$};

		    \draw(-3,6) node(z){\small Agent $1$};
		                 \draw(-3,4) node(z){\small Agent $2$};
		                 \draw(-3,2) node(z){\small Agent $3$};

		\draw(5,6.8) node(z){\small $h_1$};

		\draw(5,4.8) node(z){\small $h_2$};

		\draw(5,2.8) node(z){\small $h_2$};

		\draw(12.5,6.8) node(z){\small $h_1$};

		\draw(12.5,4.8) node(z){\small $h_1$};

		\draw(12.5,2.8) node(z){\small $h_3$};

		\draw(17.5,6.8) node(z){\small $h_3$};

		\draw(17.5,4.8) node(z){\small $h_3$};

		\draw(17.5,2.8) node(z){\small $h_3$};
		  \end{tikzpicture}
		\end{center}

		\noindent
			The final allocation computed by PS is
		\[
		PS(\succ_1,\succ_2,\succ_3)=\begin{pmatrix}
			3/4&0&1/4\\
		  1/4&1/2& 1/4 \\
		  0&1/2 &  1/2

		 \end{pmatrix}.
		\]

		\end{example}

		%
		%
		%
		%
		%
		%
		%
		%
		%

			Consider the assignment problem in Example~\ref{example:PS}. If agent $1$ misreports his preferences as follows: $\succ_1':\quad h_2,h_1,h_3,$ then \[
		PS(\succ_1',\succ_2,\succ_3)=\begin{pmatrix}
			1/2&1/3&1/6\\
		  1/2 & 1/3 & 1/6 \\
		  0 & 1/3 & 2/3
			\end{pmatrix}.
		\]
		\noindent
		Then, if $u_1(h_1)=7$, $u_1(h_2)=6$, and $u_1(h_3)=0$, then agent $1$ gets more expected utility when he reports $\succ_1'$. In the example, although truth-telling is a DL best response, it is not necessarily  an EU best response for agent $1$.

		Examples 1 and 2 of \citet{Koji09a} show that manipulating the PS mechanism can lead to an SD improvement when each agent can be allocated more than one house. In light of the fact that the PS rule can be manipulated, we examine the complexity of a single agent computing a manipulation, in other words, the best response for the PS rule.\footnote{Note that if an agent is risk-averse and does not have information about the other agent's preferences, then his maximin strategy is to be truthful. The reason is that if all agents have the same preferences, then the optimal strategy is to be truthful.} 
		For a preference profile $\pref$, we denote by $(\pref_{-i},\pref_i')$ the preference profile obtained from $\pref$ by replacing agent $i$'s preference by $\pref_i'$.
		For $\mathcal{E}\in \{SD, EU, DL\}$, we define the problem \textsc{$\mathcal{E}$-BR}: Given $(N,H,\pref)$,
				 compute a preference $\pref_1'$ for agent $1$  such that there exists no preference $\pref_1''$ such that $PS(N,H,(\pref_1'',\pref_{-1})) \succ_1^{\mathcal{E}} PS(N,H,(\pref_1',\pref_{-1}))$.
				 
		For a constant $m$, the problem \textsc{$\mathcal{E}$-BR} can be solved by brute force by trying out each of the $m!$ preferences. Hence we will not assume that $m$ is a constant.

		We establish some more notation and terminology for the rest of the paper.
		We will often refer to the PS outcomes for partial lists of houses and preferences.
		We will denote by $PS(\pref_i^{L}, \pref_{-i})(i)$, the allocation that agent $i$ receives when
		his preference is according to ordered list $L$.
		Note that preferences and ordered lists are interchangeable, except that a list need not contain all houses in $H$.
		When an agent runs out of houses
		in his preference list, he stops eating.
		The \emph{length} of a list $L$ is denoted $|L|$, and we refer to the $k$th house in $L$ as $L(k)$.
		In the PS rule, the \emph{eating start time} of a house is the time point at which the house starts to be eaten by some agent. In Example~\ref{example:PS}, the eating start times of $h_1,h_2$ and $h_3$ are $0,0$ and $0.5$, respectively.

		\section{Lexicographic best response}
		\label{sec:dl}

		In this section, we present a polynomial-time algorithm for \textsc{DL-BR}. Lexicographic preferences are well-established in the assignment literature~\citep[see \eg][]{SaSe13b,ScVa12a,Cho12a}. 
	Let $(N,H,\pref)$ be an assignment problem where $N=\{1,\dots,n\}$ and $H=\{h_1,\dots,h_m\}$. We will show how to compute a DL best response for agent $1\in N$. It has been shown that when $m\leq n$, then truth-telling is the DL best response but if $m>n$, then this need not be the case~\citep{SaSe13b,ScVa12a,Koji09a}.

		Recall that a preference $\succ_1'$ is a DL best response for agent 1 if the fractional allocation agent 1 receives by reporting $\succ_1'$ is DL preferred to any fractional allocation agent 1 receives by reporting another preference.
		That is, there is no preference $\succ_1''$ such that his share of a house $h$ when reporting $\succ_1''$ is strictly larger than when reporting $\succ_1'$ while the share of all houses he prefers to $h$ (according to his true preference $\succ_1$) is the same whether reporting $\succ_1'$ or $\succ_1''$.

		Our algorithm will iteratively construct a partial preference list for the $i$ most preferred houses of agent 1.
		Without loss of generality, denote
		 $\succ_1: h_1, h_2, \dots, h_m.$

		For any $i, 1\le i\le m$, denote $H_i = \{h_1, \dots, h_i\}$.
		A preference of agent 1 \emph{restricted} to $H_i$ is a preference over a subset of $H_i$.
		For the preference of agent 1 restricted to $H_i$, the PS rule computes an allocation where the preference of agent 1 is replaced with this preference and the preferences of all other agents remain unchanged.
		The notions of DL best response and DL preferred fractional assignments with respect to a subset of houses $H_i$ are defined accordingly for restricted preferences of agent 1.
		
			\begin{algorithm}[h!]
				\caption{DL best response for $n$ agents}
				\label{algo:DL-BR}
				\renewcommand{\algorithmicrequire}{\wordbox[l]{\textbf{Input}:}{\textbf{Output}:}}
				\renewcommand{\algorithmicensure}{\wordbox[l]{\textbf{Output}:}{\textbf{Output}:}}
				\renewcommand{\algorithmiccomment}[1]{\hfill // #1}
	 \scriptsize 
				\begin{algorithmic}
					\REQUIRE $(N,H,\succ)$ 
					\ENSURE DL Best response of agent $1$
				\end{algorithmic}
				\algsetup{linenodelimiter=\,}
				\begin{algorithmic}
					 \scriptsize 
					\STATE $L_1 \leftarrow h_1$ \COMMENT{Best response for agent 1 w.r.t. $H_1 = \{h_1\}$}
					\FOR[Compute a best response w.r.t. $H_2, \dots, H_n$]{$i=2$ to $n$}
					  \STATE $p \leftarrow 0$ 	
					  \IF{$\exists q\in \{1,\dots,i-1\}$ such that $0<PS1(L_{i-1},L_{i-1}(q))<1$}
					    \STATE $p \leftarrow \max\{q\in \{1,\dots,i-1\} : 0<PS1(L_{i-1},L_{i-1}(q))<1\}$
					  \ENDIF
					  \FOR[New house $h_i$ inserted after position $p$]{$q \leftarrow p+1$ to $|L_i|+1$}
					    \STATE $L_i^q \leftarrow L_{i-1}(1) \oplus \dots \oplus L_{i-1}(q-1) \oplus h_i$
					    \WHILE[Complete the list according to the stingy ordering]{$|L_i^q| \le |L_{i-1}|$}
					      \STATE $est \leftarrow$ \textbf{EST}$(N,H,(L_i^q, \succ_2, \dots, \succ_n))$
					      \STATE $S \leftarrow \{h\in L_{i-1} \setminus L_i^q : est(h) \text{ is minimum}\}$
					      \STATE $h_s \leftarrow $ first house among $S$ in $\succ_1$
					      \STATE $L_i^q \leftarrow L_i^q \oplus h_s$
					    \ENDWHILE
					    \IF{$PS1(L_i^q,h_i)=0$}
					      \STATE $L_i^q \leftarrow L_{i-1}$
					    \ENDIF
					  \ENDFOR				  
					  \STATE $q\leftarrow p$
					  \COMMENT{Determine which $L_i^q$ is stingy}
					  \STATE $\mathsf{worse}[p-1] \leftarrow \mathsf{true}$
					  \STATE $\mathsf{finished} \leftarrow \mathsf{false}$
					  \WHILE{$\mathsf{finished}=\mathsf{false}$}
					    \IF{$\exists h\in H_{i-1}$ such that $PS1(L_i^q,h) \ne PS1(L_{i-1},h)$}
					      \STATE $\mathsf{worse}[q] \leftarrow \mathsf{true}$
					      \STATE $q \leftarrow q+1$
					    \ELSE
					      \STATE $\mathsf{worse}[q] \leftarrow \mathsf{false}$
					      \IF{$PS1(L_i^q,h_1)>0$ \AND $PS1(L_i^q,h_1)<1$}
					        \IF{$\mathsf{worse}[q-1] = \mathsf{false}$}
					          \STATE $q \leftarrow q-1$
					        \ENDIF
					        \STATE $\mathsf{finished} \leftarrow \mathsf{true}$
					      \ELSIF{$PS1(L_i^q,h_1)=1$}
					        \STATE $est \leftarrow$ \textbf{EST}$(N,H,(L_i^q(1) \oplus \dots \oplus L_i^q(q-1), \succ_2, \dots, \succ_n))$
					        \IF{$\exists h\in \{L_i^q(q+1), \dots, L_i^q(|L_i^q|)\}$ such that $est(h)\le est(h_i)$}
					          \STATE $q \leftarrow q+1$
					        \ELSE
					          \STATE $\mathsf{finished}=\mathsf{true}$
					        \ENDIF
					      \ENDIF
					    \ENDIF
					  \ENDWHILE
					  \STATE $L_i \leftarrow L_i^q$
					\ENDFOR
					\RETURN $L_n$
				\end{algorithmic}
				\end{algorithm}
				

		For a house $h\in H$, let $PS1(L,h) = (PS(\pref_1^L,\pref_{-1})(1))(h)$ denote the fraction of house $h$ that the PS rule assigns to agent 1 when he reports the (partial) preference $L$.
	We start with a simple lemma showing that a DL best response for agent 1 for the whole set $H$ can be no better and no worse on $H_i$ than a DL best response for $H_i$.

		\begin{lemma}\label{lem:eq-ass}
		 Let $i\in\{1,\dots,m\}$. A DL best response for agent 1 on $H$ gives the same fractional assignment to the houses in $H_i$ as a DL best response for agent 1 on $H_i$.
		\end{lemma}
		%

		\noindent
		Our algorithm will compute a list $L_i$ such that $L_i \subseteq H_i$.\footnote{When we treat a list as a set we refer to the set of all elements occurring in the list.}
		The list $L_i$ will be a DL best response for agent 1 with respect to $H_i$.
		Suppose the algorithm has computed $L_{i-1}$.
		Then, when considering $H_i = H_{i-1} \cup \{h_i\}$, it needs to make sure that the new fractional allocation restricted to the houses in $H_{i-1}$ remains the same (due to Lemma \ref{lem:eq-ass}).
		For the preference to be optimal with respect to $H_i$, the algorithm needs to maximize the fractional allocation of $h_i$ to agent 1 under the previous constraint.


		Our algorithm will compute a canonical DL best response that has several additional properties.
		 A preference $L_i$ for $H_i$ is \emph{no-$0$} if $L_i$ contains no house $h$ with $PS1(L_i,h)=0$.
		Any DL best response for agent $1$ for $H_i$ can be converted into a no-$0$ DL best response by removing the houses for which agent 1 obtains a fraction of $0$.
		 For a no-$0$ preference $L_i$ for $H_i$,
		 the \emph{stingy ordering} for a position $j$
		 is determined by running the PS rule with the preference $L_i(1) \oplus \dots \oplus L_i(j-1)$ for agent 1, where $\oplus$ denotes concatenation.
		 It orders the houses from $\bigcup_{k=j}^{|L_i|} L_i(k)$ by increasing eating start times, and when two houses $h,h'$ have the same eating start time, we order $h$ before $h'$ iff $h\succ_1 h'$.
Intuitively, houses occurring early in this ordering are the most threatened by the other agents at the time point when agent 1 comes to position $j$.
		The following definition takes into account that the eating start times of later houses may change depending on agent 1's ordering of earlier houses.

		 A preference $L_i$ for $H_i$ is \emph{stingy} if it is a no-$0$ DL best response for agent 1 on $H_i$, and
		 for every $j\in\{1,\dots,i\}$, $L_i(j)$ is the first house in the stingy ordering for this position such that there exists a DL best response starting with $L_i(1) \oplus \dots \oplus L_i(j)$.
		We note that, due to Lemma \ref{lem:eq-ass}, there is a unique stingy preference for each $H_i$.

		\begin{example}
		 Consider the following assignment problem.
		 \begin{align*}
			\succ_1: h_1, h_2, h_3, h_4, h_5, h_6&&
			\succ_2: h_3, h_6, h_4, h_5, h_1, h_2
		 \end{align*}
		 The preferences $h_3, h_1, h_4, h_2$ and $h_3, h_2, h_4, h_1$ are both no-$0$ DL best responses for agent 1 with respect to $H_4$, allocating $p(1)(h_1) = 1, p(1)(h_2)=1, p(1)(h_3) = 1/2, p(1)(h_4)=1/2$ to agent 1. When running the PS rule with $h_3$ as the preference list, $h_4$'s eating start time comes first among $\{h_1,h_2,h_4\}$. However, there is no DL best response for $H_4$ starting with $h_3,h_4$. The next house in the stingy ordering is $h_1$. The preference $h_3, h_1, h_4, h_2$ is the stingy preference for $H_4$.
		\end{example}

		\noindent
		The next lemma shows that when agent 1 receives a house partially (a fraction different from 0 and 1) in a DL best response, a stingy preference would not order a less preferred house before that house.

		\begin{lemma}\label{lem:before-half}
		 Let $L_{i}$ be a stingy preference for $H_i$.
		 Suppose there is a $h_j\in H_i$ such that $0<PS1(L_{i},h_j)<1$.
		 Then, $P \subseteq H_j$, where $L_i = P \oplus h_j \oplus S$.
		\end{lemma}

		\noindent
		The next lemma shows how the houses allocated completely to agent 1 are ordered in a stingy preference.

		\begin{lemma}\label{lem:ones-same-order}
		 Let $L_i$ be a stingy preference for $H_i$.
		 If $h_j,h_k \in H_i$ are two houses such that
		 $PS1(L_i,h_j)=PS1(L_i,h_k)=1$, with $L_i = P\oplus h_j \oplus M \oplus h_k \oplus S$, then either the eating start time of $h_j$ is smaller than $h_k$'s eating start time when agent 1 reports $P$, or it is the same and $h_j \succ_1 h_k$.
		\end{lemma}
		\begin{proof}
		 Suppose not.
		 But then, $L_i$ is not stingy since swapping $h_j$ and $h_k$ in $L_i$ gives the same fractional allocation to agent 1.
		\end{proof}

		\noindent
		We now show that when iterating from a set of houses $H_{i-1}$ to $H_i$, the previous solution can be reused up to the last house that agent 1 receives partially.

		\begin{lemma}\label{lem:same-prefix}
		 Let $L_{i-1}$ and $L_i$ be stingy preferences for $H_{i-1}$ and $H_i$, respectively.
		 Suppose there is a $h\in H_{i-1}$ such that $0<PS1(L_{i-1},h)<1$.
		 Then the prefixes of $L_{i-1}$ and $L_i$ coincide up to $h$.
		\end{lemma}

		\noindent
		We are now ready to describe how to obtain $L_i$ from $L_{i-1}$. See Algorithm~\ref{algo:DL-BR} for the pseudocode.
		The subroutine \textbf{EST}$(N,H,\succ)$ executes the PS rule for $(N,H,\succ)$ and for each item, records the first time point where some agent starts eating it. It returns the eating start times $est(h)$ for each house $h\in H$.

		Let $p$ be the last position in $L_{i-1}$ such that the house $L_{i-1}(p)$ is partially allocated to agent 1.
		In case agent 1 receives no house partially, set $p:=0$ and interpret $L_{i-1}(p)$ as an imaginary house before the first house of $L_{i-1}$.
		By Lemma \ref{lem:same-prefix}, we have that $L_{i-1}(s) = L_i(s)$ for all $s\le p$.
		By Lemma \ref{lem:eq-ass}, we have that the fractional assignment resulting from $L_i$ must wholly allocate all houses $L_{i-1}(p+1), \dots, L_{i-1}(|L_{i-1}|)$ to agent 1, and allocate a share of $0$ to all houses in $H_{i-1}\setminus L_{i-1}$.

		It remains to find the right ordering for $\{L_{i-1}(s): p+1\le s\le |L_{i-1}|\} \cup \{h_i\}$.
		By Lemmas \ref{lem:before-half} and \ref{lem:ones-same-order}, the prefixes of $L_{i-1}$ and $L_i$ coincide up to $h$.
		We will describe in the next paragraph how to determine the position $q$ where $h_i$ should be inserted.
		Having determined this position one may then need to re-order the subsequent houses.
		This is because inserting $h_i$ in the list may change the eating start times of the subsequent houses.
		This leads us to the following insertion procedure.
		The list $L_i^q$ obtained from $L_{i-1}$ by inserting $h_i$ at position $q$, with $p<q\le|L_i|+1$, is determined as follows.
		Start with $L_i^q := L_{i-1}(1) \oplus \dots \oplus L_{i-1}(q-1) \oplus h_i$.
		While $|L_i^q|\le |L_{i-1}|$, we append to the end of $L_i^q$ the first house among $L_{i-1}\setminus L_i^q$ in the stingy ordering for this position.
		After the while-loop terminates, run the PS rule for the resulting list $L_i^q$.
		In case we obtain that $PS1(L_i^q,h_i)=0$, we remove $h_i$ again from this list (and actually obtain $L_i^q=L_{i-1}$).

		The position $q$ where $h_i$ is inserted is determined as follows.
		Start with $q:=p$.
		We have an array $\mathsf{worse}$ keeping track of whether the lists $L_i^p , \dots, L_i^{i}$ produce a worse
		outcome for agent 1 than the list $L_{i-1}$.
		Set $\mathsf{worse}[p-1]:=\mathsf{true}$.
		As long as the list $L_i$ has not been determined, proceed as follows.
		Obtain $L_i^q$ from $L_{i-1}$ by inserting $h_i$ at position $q$, as described earlier.
		Consider the allocation of agent 1 when he reports $L_i^q$.
		If this allocation is not the same for the houses in $H_{i-1}$ as when reporting $L_{i-1}$, then set $\mathsf{worse}[q] := \mathsf{true}$, otherwise set $\mathsf{worse}[q] := \mathsf{false}$.
		If $\mathsf{worse}[q]$, then increment $q$.
		This is because, by Lemma \ref{lem:eq-ass}, this preference would not be a DL best response with respect to $H_i$.
		Otherwise, if $0<PS1(L_i^q,h_i)<1$, then we can determine $h_i$'s position. If $\mathsf{worse}[q-1]$, then set $L_i := L_i^q$, otherwise set $L_i:= L_i^{q-1}$.
		This position for $h_i$ is optimal since moving $h_i$ later in the list would decrease its share to agent 1.
		Otherwise, we have that $\mathsf{worse}[q] = \mathsf{false}$ and $PS1(L_i^q,h_i) \in \{0,1\}$.
		This will be the share agent 1 receives of $h_i$.
		If $PS1(L_i^q,h_i) = 0$, then set $L_i:=L_{i-1}$.
		Otherwise ($PS1(L_i^q,h_i) = 1$), it still remains to check whether the current position for $h_i$ gives a stingy preference.
		For this, run the PS rule with the preference $L_i^q(1) \oplus \dots \oplus L_i^q(q-1)$ for agent 1. If $h_i$'s eating start time is smaller than the eating start time of each house $L_i^q(r)$ with $r>q$, then set $L_i := L_i^q$, otherwise increment $q$.

		Thus, given $L_{i-1}$, the preference $L_i$ can be computed by executing the PS rule $O(m)$ times.
		The DL best response computed by the algorithm is $L_m$.
		Since the PS rule can be implemented to run in linear time $O(nm)$, the running time of this DL best response algorithm is $O(nm^3)$.

		\begin{theorem}\label{th:dlbr-in-P}
		 \textsc{DL-BR} can be solved in $O(nm^3)$ time.
		\end{theorem}

		\begin{example}
		 Consider the following instance. 
		 \begin{align*}
			\succ_1: h_1, h_2, h_3, h_4, h_5, h_6, h_7, h_8, h_9, h_{10}\\
			\succ_2: h_8, h_3, h_5, h_2, h_{10}, h_1, h_6, h_7, h_4, h_9\\
			\succ_3: h_9, h_4, h_7, h_1, h_2, h_6, h_5, h_3, h_8, h_{10}
		 \end{align*}
		 After having computed $L_2=h_1, h_2$, the algorithm is now to consider $H_3$. Since $PS1(L_2,h_1)=PS1(L_2,h_2)=1$, the algorithm first considers $L_3^1 = h_3, h_2, h_1$. Note that $h_1$ and $h_2$ have been swapped with respect to $L_2$ since agent 2 starts eating $h_2$ before agent 3 starts eating $h_1$ when agent 1 reports the preference list consisting of only $h_3$. It turns out that $PS1(L_3^1,h_1)=PS1(L_3^1,h_2)=PS1(L_3^1,h_3)=1$. Thus, $\mathsf{worse}[1] = \textsf{false}$. Since $h_3$ does not come first in the stingy ordering, the algorithm needs to verify whether moving $h_3$ later will still give a DL best response with respect to $H_3$.
		 It then considers $L_3^2 = h_1, h_3, h_2$. However, this allocates only half of $h_3$ to agent 1, implying $\mathsf{worse}[2] = \mathsf{true}$. Since $\mathsf{worse}[1] = \mathsf{false}$, the algorithm sets $L_3 = L_3^1$.
		 The DL best response computed by the algorithm is $L_{10}=h_3, h_2, h_1, h_6$.
		\end{example}

		\noindent
		We note that a DL best response is also an SD best response. 
		One may wonder whether an algorithm to compute the DL best response also provides us with an algorithm to compute an EU best response. However, a DL best response may not be an EU best response for three or more agents.
			Consider the preference profile in Example~\ref{example:PS}. Since the number of houses is equal to the number of agents, reporting the truthful preference is a DL best response~\citep{ScVa12a}. However, we have shown a different preference for agent 1 where he may obtain higher utility.




\section{Expected utility best response}

		In this section, we consider the problem of expected utility best response.
			 \subsection{Case of two agents}
		We first show that for the case of two agents, an EU best response can be computed in linear time. The result hinges on a close connection that we identify between PS and discrete allocation of objects to agents via \emph{sequential allocation}. 
		In the sequential allocation setting $(N,O,\pref',\pi)$, there is an agent set $N$, an object set $O=\{o_1,\ldots o_{m'}\}$, a preference profile $\pref'$ that specifies for each agent $i\in N$ his preferences $\pref_i'$ over $O$, 
		and a policy $\pi:\{1,\ldots, m'\}\rightarrow N$. The sequential allocation rule works as  follows. Starting from $j=1$ to $m'$, agent $\pi(j)$ gets his most preferred object that is not yet allocated. If no unallocated object is on the preference list of the agent, then the agent does not get any object when his turn comes. The assignment as a result of sequential allocation is denoted by $SA(N,O,\pref',\pi)$.
		We will restrict ourselves to the case where $N=\{1,2\}$ and will only consider the alternating policy $\pi^*=1212\ldots$ in which agent $1$ starts first and then the agents keep alternating. The sequential allocation setting was introduced by 
		~\citet{KoCh71a} where they showed that the best response can be computed in linear time when $|N|=2$ and the policy is the alternating sequence. Recently, \citet{BoLa14a} generalized their result to the case of any number of agents, any policy, and where the manipulator may be indifferent between objects.

			We highlight a close connection between sequential allocation and PS and thereby between allocation mechanisms for indivisible and divisible houses.
		For the random assignment setting $(\{1,2\},H,\pref)$, the \emph{half-house reduction} gives us the sequential allocation setting $(\{1,2\},O,\pref',\pi^*)$. In the reduction, each house $h_j\in H$ is cloned so that we have two half-houses $h_j^1$ and $h_j^2$ for each house $h_j$:  $O=\{h_j^1,h_j^2\midd j=1,\ldots, m\}$. 
%
		Both agents have preferences over half-houses that are consistent with their preferences over houses and for each house, each agent prefers the first half-house slightly more than the second half-house: if $h_j\pref_i h_k$, then $h_j^{1}\pref_i' h_j^{2}\pref_i'h_k^{1}\pref_i' h_k^2$.
		We show that for $n=2$, the assignment under PS is `essentially' the same as the assignment obtained by applying sequential allocation to the setting resulting from the half-house reduction:

		\begin{remark}\label{remark:half-house}
			The assignment $PS(\{1,2\},H,\pref)$ and the assignment  $SA(\{1,2\},O,\pref',\pi^*)$ are related as follows:		
			$PS(\{1,2\},H,\pref)(i)(h_j)= \frac{1}{2}\cdot (SA(\{1,2\},O,\pref',\pi^*)(i)(h_j^1) + SA(\{1,2\},O,\pref',\pi^*)(i)(h_j^2))$. 

	

			\end{remark}
	
			We note that in the half-house reduction, each preference list $\pref_i'$ satisfies the \emph{consecutivity property}: half-houses corresponding to the same house are placed consecutively in the preference list. We will use the consecutivity property in our argument.

		\begin{theorem}\label{th:2agents-eubr-in-P}
			For the case of two agents, an EU best response can be computed in linear time.
		\end{theorem}
		\begin{proof}
			We consider the EU best response problem for PS where  the manipulator, agent $1$, has preferences $\pref_1: h_1,\ldots, h_m$. 
			The main idea is to reduce the EU best response problem $(\{1,2\},H,\pref)$ for PS to the EU best response problem $(\{1,2\},O,\succsim',\pi^*)$ for sequential allocation. The reduction is a slight modification of the half-house reduction with the difference that agent $1$ is indifferent between two half-houses corresponding to the same house. The object set is $O=\{h_j^1,h_j^2\midd j=1,\ldots, m\}$. 
		In  $\succsim'$, both agents have preferences over half-houses that are consistent with their preferences over houses. 
		We will assume without loss of generality that agent $2$ prefers the first half-house slightly more than the second half-house. Agent $1$ is indifferent between any two half-houses corresponding to the same house: $h_j^1\sim_1' h_j^2$ for all $j\in \{1,\ldots, m\}$ but will be required to report strict preferences.  When we consider sequential allocation, we will view it in rounds so that in each round, first agent $1$ picks a most preferred available house and then agent $2$ picks a most preferred available house.


		In the algorithm by \citet{BoLa14a}, when agents have strict preferences, it is checked whether the manipulator (agent $1$) can get different target sets of objects. In the algorithm, only a linear number of target sets need to be considered.
		Given target set $T_k$ which is restricted to objects from $o_1,\ldots, o_k$, we can compute target set $T_{k+1}$ as follows: check whether target set $T_{k}\cup \{o_{k+1}\}$ can be achieved or not.  $T_{k+1}=T_{k}\cup \{o_{k+1}\}$  if $T_{k}\cup \{o_{k+1}\}$ can be achieved and $T_{k+1}=T_{k}$ otherwise. $T_m$ is then the most preferred allocation that agent $1$ achieves and the allocation is unique. 
		When the manipulator is indifferent among objects, \citet{BoLa14a} showed that their algorithm can be easily modified as follows: agent $1$ considers a linear order instead of his actual weak order where the linear order is achieved by breaking ties between the indifferent objects in the same order as the preference of agent $2$. Based on this insight, observe that both agents will pick $h_j^1$ before $h_j^2$ for any $j\in \{1,\ldots, m\}$ if they report truthfully.


		We first show that there exists a best response of agent $1$ in the sequential allocation setting $(N,O,\succsim',\pi^*)$ that satisfies the consecutivity property.
		If agent $1$ either gets both half-houses corresponding to a house or none of them, then his optimal preference report for sequential allocation  trivially satisfies the consecutivity property.
		If this is not the case, then let us consider the most preferred house $h_j$ for which agent $1$ gets one of the corresponding half-houses but not the other. 
If agent $1$ only gets $h_j^1$ but not $h_j^2$, this means that in his best response for houses restricted to $\{h_1^1, \ldots, h_j^2\}$, $h_j^2$ was already taken by agent $2$ in a round in which agent $1$ picked some other object. Then agent $1$ can eventually insert $h_j^2$ immediately after $h_j^1$ in his best response preference knowing well that he will not get $h_j^2$. Thus, the best response for sequential allocation can be modified so that it satisfies the consecutivity property and  yields the same optimal allocation. 
Now consider the case where agent $2$ gets $h_j^1$ but agent $1$ gets $h_j^2$.
Then this means that agent $1$ cannot get $h_j^1$ in his best response when his preference is restricted only to houses from the set $\{h_1^1,h_1^2,\ldots, h_{j-1}^1, h_{j-1}^2, h_j^1\}$. Therefore, agent $1$ can still insert $h_j^1$ eventually just before $h_j^2$ in his best response so that the consecutivity property is satisfied and the allocation does not change even though agent $1$ does not get $h_j^1$ in his best response.


		We now show that the best response of agent $1$ in the sequential allocation setting $(N,O,\succsim',\pi^*)$ can be used to compute the best response of agent $1$ in $(N,H,\pref)$ under PS.
		Let $U$ be the expected utility for agent $1$ under his best response $\pref_1^*$ in the PS setting. The best response $\pref_1^*$ corresponds to ${\pref_1^*}'$ over the set of half-houses. 
		By Remark~\ref{remark:half-house}, agent $1$ achieves essentially the same allocation and hence the same utility $U$ in the sequential allocation setting if he submits preference ${\pref_1^*}'$.
		Conversely, if agent $1$ achieves utility $U$ in the sequential allocation setting via a preference report, then he achieves at least as much utility by reporting his optimal preference ${\pref_1^*}'$ constructed via the algorithm of \citet{BoLa14a}. Hence, the preference ${\pref_1^*}'$ can be modified as shown above so that it satisfies the consecutivity property. In this case,
		there exists a preference $\pref_1^*$ over $H$ which is consistent with the preferences ${\pref_1^*}'$ over $O$. If agent $1$ reports $\pref_1^*$, then he gets essentially the same allocation as $SA(\{1,2\}, O, ({\pref_1^*}', \pref_2')(1)$ and thus gets utility $U$.
			\end{proof}

		The best response algorithm of \citet{BoLa14a} returns the same optimal preference report for all cardinal utilities consistent with the ordinal preference of the manipulator. Next, we point out that for the case of two agents and the PS rule, a DL best response and an EU best response are equivalent.

				\begin{proposition}
					For the case of two agents and the PS rule, a DL best response is an EU best response and an EU best response is a DL best response.
				\end{proposition}
				\begin{proof}
					For two agents, PS assigns probabilities from the set $\{0,1/2,1\}$. Hence DL preferences can be represented by EU preferences where the utilities are exponential: the utility of a more preferred house is twice the utility of the next most preferred house. Hence a response is a DL best response if it is an EU best response for exponential utilities. 
					On the other hand we have shown that for two agents and the PS rule, an EU best response is the same for any utilities compatible with the preferences. Hence for two agents, an EU best response for any utilities is the same as the EU best response for exponential utilities which in turn is the same as a DL best response.
				\end{proof}
		
		
			\subsection{General case}
		
		We show that an EU best response is NP-hard to compute. The result contrasts with Theorem~\ref{th:dlbr-in-P} which states that a DL best response can be computed in polynomial time. 
%

	%
	%
		
		%
		%
		
		%
		\begin{theorem}\label{th:eubr-nphard}
			 \textsc{EU-BR} is NP-hard.
		\end{theorem}

		%

		\begin{proof}
		To show hardness we show that the following problem is NP-complete: given an assignment setting as well as a utility function $u:H\rightarrow \mathbb{N}$ specifying the utility of each house for the manipulator (agent $1$) and a target utility $T$, can the manipulator specify preferences such that the utility for his allocation under the PS rule is at least $T$? We reduce from a restricted NP-hard version of 3SAT where each literal appears exactly twice in the formula. Given such a 3SAT instance $F=(X,C)$ where $X = \{x_1,\dots,x_{n}\}$ is the set of variables and $C$ the set of clauses, we build an instance of \textsc{EU-BR} where the manipulator can obtain utility $\geq T$ if and only if the formula is satisfiable.  At a high level, we will create an instance of the assignment problem which can be conceptualized as 18 (mostly) disjoint parts that we index by $D \in \{1, \ldots, 18\}$.  We will describe the main (first) part in detail and explain how it is duplicated to create the other 17 parts. Each of the 18 parts is divided up into $n$ choice rounds which we index from $1$ to $n$.  For each part there is an additional \emph{clause} round. The 18 parts are linked by a special set of houses which allow us to synchronize the timing of the manipulator with respect to all the other agents. 
		The set of agents is
		$N= \{1\} \cup \bigcup_{D=2}^{18} \{a^D_{\text{dummy}}\} \cup \bigcup_{D=1}^{18} A^D_{\text{literals}}$
		where
		the manipulator is Agent 1, 17 `dummy' manipulators for the 17 copies of the main part, and two agents for each positive and negative literal in the formula for each of the 18 parts, $A^D_{\text{literals}} = \{a^{1,D}_{x_i}, a^{2,D}_{x_i}, a^{1,D}_{\neg x_i}, a^{2,D}_{\neg x_i} : x_i\in X\}$.


		%

		The set of houses is
				$H =  H_{\text{slow}} \cup \bigcup_{D=1}^{18} H^D_{\text{rounds}} \cup \bigcup_{D=1}^{18} H^D_{\text{clause}} \cup \bigcup_{D=2}^{18}\{h^D_{CP}\} \cup \{h_{\text{prize}}\}$ where 
		$H_{\text{slow}} = \{ h^r_s: r \in \{1, \ldots, n-1\} \}$ is the set of slowdown houses that will be used to control the timing of the manipulator's decisions. Note that there is only one slowdown house per round and these houses are shared between all $18$ parts. $H^D_{\text{rounds}} = \{h^{r,D}_{x_i}, h^{r,D}_{\neg x_i}: r \in \{1, \ldots, n\}, i \in \{1, \ldots, n\}\}$ is a set of houses consisting of one house for each positive and negative literal in the formula for each of the $n$ rounds; $H_{\text{clause}} = \{h^{1,D}_{c}, h^{2,D}_{c}, h^{3,D}_{c}: c \in \{1, \ldots, C\} \}$ is a triplet of houses for each clause in the formula; $h_{\text{prize}}$ is the prize house for the manipulator;  and $\bigcup_{D=2}^{18}\{h^D_{CP}\}$ is the set of consolation prize houses for the dummy manipulators.

We will describe how to construct the preferences for the main part which contains the manipulator, agent 1, and then explain the small differences necessary to create the 17 other duplicate instances.  Example~\ref{example:sat} gives an illustration of the main part of a small instance and may be helpful for reference during the discussion.

		\paragraph{Main part}
		
		We will describe the rounds by declaring which houses are eaten in them and show how the preference lists of the agents are constructed. Each agent's preference list can be described has having a \emph{head} and a \emph{tail}. To ease the description, we will omit the round index $D=1$ in the variable names. Intuitively, the head consists of the houses that the agent will consume during the running of the PS algorithm while the tail consists of houses that will not be eaten. When we describe how we add houses to an agent's preference list, we will say \emph{append the house(s) to the head} to mean add this set of houses to the end of the head of the preference list, behind those that have been placed before.  We say \emph{append the house(s) to the tail} of the preferences to mean place them last amongst all houses which have been placed in the preferences so far.
		
	In each choice round $r$, houses $h^r_{x_i}$ and $h^r_{\neg{x_i}}$ for each $i \in \{1, \ldots, n\}$ will be eaten. Append those houses to the head of the preferences of the agents corresponding to the same literal and append them to the tail of the preferences of agents associated to a different literal. Append houses $h^r_{x_r}$ and $h^r_{\neg{x_r}}$ to the head of the manipulator's preferences (the order in which we add them in is not important). Houses $h^r_{x_i}$ and $h^r_{\neg{x_i}}$ where $i \neq r$ are appended to the tail of the manipulator's preferences. In each choice round except the last one, slowdown house $h^r_s$ will be eaten. We append it to the tail of the preferences of the literal agents, and to the head of the preferences of the manipulator agent (after the literal houses we added for this round).
		
		Finally we describe the clause round. For each clause, we have the 3 houses $h^1_c, h^2_c, h^3_c$. We append these 3 houses to the head of the preferences of exactly 1 agent corresponding to the negation of each of the clause $c$'s literals. If an agent has already had houses added to his preferences in the clause round, we add them to the other agent corresponding to the same literal (since a literal appears only twice in the formula, this ensures each agent has only one triplet of houses appended to the head of their preferences).
		The prize house $h_{\text{prize}}$ is appended to the head of both the manipulator's and the literal agents' preferences (after the clause houses we just added to the literal agents).
		
		\paragraph{Duplicate parts} For each of the duplicate parts, $D\in \{2, \ldots, 18\}$, we will describe the necessary modifications. For clarity we call the copy of the prize house in the duplicated parts of the instance consolation prize houses denoted $h^D_{CP}$ for each $D \in \{2, \ldots, 18\}$. 
Recall that the set of slowdown houses $H_{\text{slow}}$ is shared between all the parts; thus all the parallel constructions `merge' at the set of slowdown houses. We are left with the fact that houses from a given duplicate part $D$ of the instance have not been added to the preferences of agents from all other parts of the instance. We can append all these houses to the tail of the preferences of the agents outside this part in any order.

		\paragraph{The manipulator's utilities} We will give the manipulator's utility in terms of a number $\alpha$ to be fixed later. 
		The prize house has utility 1. The literal houses that are appended to the head of the manipulator's preferences during round $i$ ($h^i_{x_i}$ and $h^i_{\neg{x_i}}$) have utility $(2\alpha)^{2(n-i)}$ and $(2\alpha)^{2(n-i)}+\epsilon$ where $\epsilon$ is $O(\frac{1}{2^{n}})$. 
The slowdown houses have utility $(2\alpha)^{2(n-i-1)+1}$. All other houses have negligible utility. By negligible we mean that adding up all their combined utilities will yield less than $\frac{1}{\alpha}$ utility. This can be done since we have a polynomial number of houses and we can make the utilities exponentially small. 

		Based on these utilities we can now derive a target value for $T$ and analyse the behaviour that the manipulator must have to reach that target.
		The manipulator may only start eating a new house once the house he is currently eating is no longer available. This means that if he starts eating a house, he is `stuck' eating said house for a certain amount of time.
		We now constrain the manipulator's possibilities by showing that by diverging from the literal and slowdown houses he should be eating according to his preferences, he will commit to a house for which he has exponentially less utility for an amount of time which is at least some constant.
		By setting $\alpha$ to be large enough, we can ensure that this loss in utility is irrecuperable. We say the manipulator \emph{behaves as prescribed} if he declares preferences which correspond to his true preferences up to permutations of the literals associated with one same variable.
		
		Let $t_1>0$ be the smallest amount of time the manipulator will eat a new house if he has behaved as prescribed in all his previous choices. The next lemma shows that $t_1$ is independent of the instance size.
		

		\begin{lemma}\label{lemma:sametime}
		$t_1 \in O(1)$.
		\end{lemma}
		\begin{proof}
		As the algorithm progresses, we may group the agents in a constant number of groups based on the extent they have eaten their current house when the manipulator finishes consuming one of his houses and the number of agents eating that house. Each group is associated with a value, which corresponds to the amount of time the manipulator would have to spend if he decided to eat a house currently being eaten by members of that group. By showing that the number of these groups is constant, and therefore so is the number of values, we show that $t_1$ is a constant.
		The groups can be characterized by the type of house that the members are eating. At any point in the algorithm we say that a literal has been chosen by the manipulator if the round $r$ is greater than the index $i$ of that literal, $r > i$. We say that a literal is untouched by the manipulator for $i> r$.  The groups are defined as follows:
		%
		%
		\begin{inparaenum}[1)] 
		\item Agents eating houses being eaten by an agent corresponding to a literal which has been chosen by the manipulator
		\item Agents eating houses being eaten by an agent corresponding to a literal which is the negation of one chosen by the manipulator
		\item Agents eating houses corresponding to literals untouched by the manipulator
		\item Agents eating houses being eaten by dummy manipulators.
		\end{inparaenum}

		At the start of any round $i$, eating a house from group $j$ would take $g^j_1$ time. The manipulator then finishes eating the first literal, and eating a house from group $j$ would take $g^j_2$ time. After eating the second literal, eating a house from group $j$ would take $g^j_3$ time. Finally the manipulator eats the slowdown houses and we have corresponding value $g^j_4$. We will now show that the values for $g^j_l$ are the same for all rounds. To show this we simply need to make sure that all the agents stay `synchronised'. It takes the manipulator 0.5 units of time to finish the current round ($\frac{1}{3}$ on the first literal, $\frac{1}{9}$ on the second, and $\frac{1}{18}$ on the slowdown house).
		Let us now show that it also takes 0.5 units of time for every other group to get to the same point in the next round. The exception are the agents eating a house that is also being eaten by the manipulator or some dummy in that round, which fall out of sync with their previous group (group $3$ or $4$) and transit either to group $1$ or $2$. 
For groups 1-3, all these agents pair up and have $1$ house per round. It therefore takes them each 0.5 time to eat it. 
For Group 4, the dummy manipulators eat a first literal ($\frac{1}{3}$) then a second ($\frac{1}{9}$) and finally all 18 manipulators  join together and eat the slowdown houses in the round, which takes them time $\frac{1}{18}$. This adds up to $\frac{9}{18}=0.5$.\end{proof}

		\begin{corollary}
		There is value for $\alpha \in O(1)$ such that the manipulator behaves as prescribed.
		\end{corollary}

		\begin{lemma}
		In the clause round all agents corresponding to literals chosen by the manipulator start the round at the same time as the manipulator, whilst agents corresponding to negation of the choice of the manipulator are in advance and start the round $\frac{1}{9}$ units of time before the manipulator.
		\end{lemma}
		\begin{proof}
		In Lemma~\ref{lemma:sametime} we argued that the agents took the same amount of time to finish a round. The exception to this is the last round where the manipulator does not eat any slowdown houses and therefore finishes the round at the same time as group $1$. Group $2$ finishes the round $\frac{1}{9}$ before group $1$ since the manipulator spent $\frac{1}{3}$ time eating a house with them whereas he spent $\frac{1}{9}$ time eating a house with agents from group 2. This results in a $\frac{4}{9}-\frac{3}{9}=\frac{1}{9}$ delay between the two.
		\end{proof}


		The manipulator's choice corresponds to an assignment of the variables in the SAT formula. If the manipulator chose to eat house $h^r_{x_r}$ before $h^r_{\neg{x_r}}$ then this corresponds to setting $x_r$ to true (and vice versa).  Thus, in each round the manipulator choses an assignment for a variable in the formula.
	 	The target utility $T$ is the sum of $\frac{4}{9}$ of the utility of $h^r_{x_r}$ and$\frac{1}{18}$ of the utility of the slowdown house $h^r_s$ (except in the last round) for each round $r$ and an extra $\frac{25}{27}$.
	 	
		\begin{lemma}
		In the clause round, the manipulator must eat the prize house before any other agent to reach the target utility $T$.
		\end{lemma}
		
		\begin{lemma}
	$F$ is satisfiable iff the manipulator can reach the target utility $T$.
		\end{lemma}
		\begin{proof}
		($\Rightarrow$) We have set $T$ so that if the manipulator declared a prescribed preference profile, he will require an extra $\frac{25}{27}-\epsilon\cdot n$ utility to reach $T$. If all clauses are satisfied, at most 2 of the agents eating the houses corresponding to a clause will be in advance and the manipulator will have $\frac{25}{27}$ units of time to eat the prize house alone.  The manipulator will always have $\frac{8}{9}$ units of time to eat the prize house alone while the other literal agents are eating the corresponding clause houses. In the worst case, 2 agents are in advance for any clause by $\frac{1}{9}$, units of time, which means that they, along with the third agent in the clause, will finish their triplet of clause houses after $\frac{8}{9} + \frac{1}{27}$ units of time, at which time all three agents will begin eating the prize house.  This leaves the manipulator to eat alone for $\frac{1}{27}$ extra time thus ensuring him extra utility $\geq \frac{25}{27}$.
		\\
		($\Leftarrow$) If the truth assignment causes a clause to be unsatisfied, the agents corresponding to the negation of the literal in the clause (and therefore eating the clause houses corresponding to the clause) will all be in advance and will finish eating the clause houses before the manipulator has eaten $\frac{25}{27}$ of the prize house. If all 3 agents are in advance, they will finish eating the clause houses $\frac{24}{27}$ units of time after the manipulator has started eating the prize house. Therefore for $\frac{3}{27}$ of the prize house there are at least 3 extra agents eating the prize house. Since this makes at least 4 agents eating $\frac{3}{27}$ of the prize house, the manipulator will get at most $\frac{1}{36}$ instead of the required $\frac{1}{27}$ of the prize house after he has eaten a share of $\frac{24}{27}$. Since the prize house is the only remaining house with non-negligible utility, and we have made $\alpha$ large enough, he cannot compensate this loss of utility by getting more of some other house.
		\end{proof}
		\end{proof}

The reduction can be used to show that even checking whether there exists any report that yields more utility than the truthful report is NP-hard. 

 \begin{example}	\label{example:sat}
	 We illustrate the reduction in the proof of Theorem~\ref{th:eubr-nphard}.
 For the following SAT formula, the table below  illustrates the preference profile for the agents in the main part. Houses not shown in the preferences are never eaten by the agents and come later in the preference lists.
 		\begin{equation*}\small \underbrace{(x_1\vee x_2\vee x_3)}_{c_1}
		\underbrace{(\neg x_1\vee \neg x_2\vee \neg x_3)}_{c_2}
			\underbrace{(x_1\vee \neg x_2\vee x_3)}_{c_3}
			\underbrace{(\neg x_1\vee x_2\vee \neg x_3)}_{c_4}
		\end{equation*}
	 	\end{example}
		
		\vspace{-1em}
	 		\begin{table}[h!]
				\centering
				\scalebox{0.65}{
	 		\begin{tabular}{ |l|l|l|l|l|l|l|l|l|l|l|l|l|l|l|l|l|l|l|l|}
	 		  \hline
			  	 	&choice round 1&&choice round 2&&choice round 3&clause round&\\  \hline
	 	$1$&$h_{x_1}^1,h_{\neg x_1}^1$&$h_s^1$&$h_{x_2}^2,h_{\neg x_2}^2$&$h_s^2$&$h_{x_3}^3,h_{\neg x_3}^3$&$h_{\text{prize}}$&\\  \hline
	 	$a_{x_1}^1$&$h_{x_1}^1$&& $h_{x_1}^2$&& $h_{x_1}^3$&$h_{c_2}^1,h_{c_2}^2,h_{c_2}^3$&$h_{\text{prize}}$\\
	 	$a_{x_1}^2$&$h_{x_1}^1$& & $h_{x_1}^2$& & $h_{x_1}^3$&$h_{c_4}^1,h_{c_4}^2,h_{c_4}^3$&$h_{\text{prize}}$\\   \hline
	 	$a_{\neg x_1}^1$&$h_{\neg x_1}^1$&& $h_{\neg x_1}^2$&& $h_{\neg x_1}^3$&$h_{c_1}^1,h_{c_1}^2,h_{c_1}^3$&$h_{\text{prize}}$\\
	     $a_{\neg x_1}^2$&$h_{\neg x_1}^1$&& $h_{\neg x_1}^2$&& $h_{\neg x_1}^3$&$h_{c_3}^1,h_{c_3}^2,h_{c_3}^3$&$h_{\text{prize}}$\\   \hline
	 	$a_{x_2}^1$&$h_{x_2}^1$&& $h_{x_2}^2$&& $h_{x_2}^3$&$h_{c_2}^1,h_{c_2}^2,h_{c_2}^3$&$h_{\text{prize}}$\\
	 	$a_{x_2}^2$&$h_{x_2}^1$&& $h_{x_2}^2$&& $h_{x_2}^3$&$h_{c_3}^1,h_{c_3}^2,h_{c_3}^3$&$h_{\text{prize}}$\\   \hline
	 	$a_{\neg x_2}^1$&$h_{\neg x_2}^1$&& $h_{\neg x_2}^2$&& $h_{\neg x_2}^3$&$h_{c_1}^1,h_{c_1}^2,h_{c_1}^3$&$h_{\text{prize}}$\\
	     $a_{\neg x_2}^2$&$h_{\neg x_2}^1$&& $h_{\neg x_2}^2$&& $h_{\neg x_2}^3$&$h_{c_4}^1,h_{c_4}^2,h_{c_4}^3$&$h_{\text{prize}}$\\   \hline
	 	$a_{x_3}^1$&$h_{x_3}^1$&& $h_{x_3}^2$&& $h_{x_3}^3$&$h_{c_2}^1,h_{c_2}^2,h_{c_2}^3$&$h_{\text{prize}}$\\
	 	$a_{x_3}^2$&$h_{x_3}^1$&& $h_{x_3}^2$&& $h_{x_3}^3$&$h_{c_4}^1,h_{c_4}^2,h_{c_4}^3$&$h_{\text{prize}}$\\   \hline
	 	$a_{\neg x_3}^1$&$h_{\neg x_3}^1$&& $h_{\neg x_3}^2$&& $h_{\neg x_3}^3$&$h_{c_1}^1,h_{c_1}^2,h_{c_1}^3$&$h_{\text{prize}}$\\
	     $a_{\neg x_3}^2$&$h_{\neg x_3}^1$&& $h_{\neg x_3}^2$&& $h_{\neg x_3}^3$&$h_{c_3}^1,h_{c_3}^2,h_{c_3}^3$&$h_{\text{prize}}$\\
	 		  \hline
	 		\end{tabular}
			}
			\label{table:psreduction}
	 		\end{table}

		\section{Conclusions}

%

		We conducted a detailed computational analysis of strategic aspects of the PS rule. 
		Since PS performs better than RSD in terms of efficiency and envy-freeness, the only drawback it has in comparison with RSD is its manipulability. We have shown that although PS is manipulable, finding an optimal manipulation is a complex task for an agent even if he has complete knowledge about the preferences of other agents. 
There is scope to conduct detailed experiments on the pecentage of instances that are manipulable and the extent and effects of manipulation. Initial experiments show that manipulation is often possible and more often decreases social welfare than increases it, though the overall effect is small. As the number of houses relative to the number of agents grows, the opportunities to manipulate increase, maximizing around 99\%.
It will be interesting to extend our results to the extension of PS for indifferences~\citep{KaSe06a}.
	Finally, studying coalitional manipulations and a deeper analysis of Nash dynamics are other interesting directions.





\subsubsection*{Acknowledgments}

NICTA is funded by the Australian Government through the Department of Communications and the Australian Research Council through the ICT Centre of Excellence Program. Serge Gaspers is the recipient of an Australian Research Council Discovery Early Career Researcher Award (project number DE120101761).
 
 \balance

\end{document}